\documentclass[sigconf]{acmart}

\usepackage{amsthm,amsmath}
\usepackage{tikz}
\usepackage{pgfplots}

\usepackage{mathtools}
\usepackage{enumerate}
\usepackage{multirow}

\newtheorem{theorem}{Theorem}

\newtheorem{thm}[theorem]{Theorem}

\newtheorem{algorithm}[theorem]{Algorithm}

\newtheorem{example}[theorem]{Example}

\def\qed{\quad\rule{1ex}{1ex}}

\makeatletter
\newcommand{\myitem}[1]{%
\item[(#1)]\protected@edef\@currentlabel{#1}%
}
\makeatother

\def\eatspace#1{#1}
\def\step#1#2{\par\kern1pt\hangindent#2em\hangafter=1\noindent\rlap{\small#1}\kern#2em\relax\eatspace}

\let\set\mathbb
\def\<#1>{\langle#1\rangle}

\def\diag{\operatorname{diag}}

\def\e{\mathrm{e}}

\copyrightyear{2022}
\acmYear{2022}
\setcopyright{acmcopyright}\acmConference[ISSAC '22]{Proceedings of the 2022 International Symposium on Symbolic and Algebraic Computation}{July 4--7, 2022}{Villeneuve-d'Ascq, France}
\acmBooktitle{Proceedings of the 2022 International Symposium on Symbolic and Algebraic Computation (ISSAC '22), July 4--7, 2022, Villeneuve-d'Ascq, France}
\acmPrice{15.00}
\acmDOI{10.1145/3476446.3535486}
\acmISBN{978-1-4503-8688-3/22/07}

\begin{document}
\fancyhead{}
\title{Guessing with Little Data}
\titlenote{M. Kauers was supported by the Austrian FWF grant P31571-N32.
  The research reported in this article has been partly funded by BMK, BMDW,
  and the Province of Upper Austria in the frame of the COMET Programme
  managed by FFG in the COMET Module S3AI.}

\author[M. Kauers]{Manuel Kauers}
\affiliation{%
  \institution{Institute for Algebra, Johannes Kepler University}
  \city{A4040 Linz}
  \country{Austria}
}
\email{manuel.kauers@jku.at}

\author[C. Koutschan]{Christoph Koutschan}
\affiliation{%
  \institution{RICAM, Austrian Academy of Sciences}
  \city{A4040 Linz}
  \country{Austria}
}
\email{christoph.koutschan@ricam.oeaw.ac.at}

\begin{abstract}
  Reconstructing a hypothetical recurrence equation from the first
  terms of an infinite sequence is a classical and well-known technique
  in experimental mathematics.
  We propose a variation of this technique which can succeed with
  fewer input terms.
\end{abstract}

\ccsdesc[500]{Computing methodologies~Algebraic algorithms}

\keywords{Experimental Mathematics, D-finite Functions, Lattice Reduction, Integer Sequences}
\maketitle

\section{Introduction}

A simple but powerful technique which has become an important tool in
experimental mathematics takes as input the first few terms of an infinite
sequence and returns as output a plausible hypothesis for a recurrence
equation that the sequence may satisfy, or a plausible hypothesis for a
differential equation satisfied by its generating function. The principle
is known as automated guessing as it somehow makes a guess how the
infinite sequence continues beyond the finitely many terms supplied as
input. In certain situations where sufficient additional information is
available about the sequence at hand, automated guessing can be combined
with other techniques from computer algebra that confirm that the guessed
equation is correct. One of many successful applications of this paradigm
is the proof of the qTSPP conjecture~\cite{koutschan10a}.

Technically, the guessing problem for linear recurrences can be solved by linear algebra. Given
$a_0,\dots,a_N$, we choose an order $r$ and a degree $d$ and make an ansatz
with undetermined coefficients: $\sum_{i=0}^r\sum_{j=0}^d c_{i,j}n^j a_{n+i}=0$.
Using the available terms of the sequence, we can instantiate
the ansatz for $n=0,\dots,N-r$ and get $N-r+1$ linear constraints on the
$(r+1)(d+1)$ unknown coefficients~$c_{i,j}$. If $r$ and $d$ are chosen
such that $(r+1)(d+2)\leq N+2$, this homogeneous linear system is not
expected to have a (nontrivial) solution. If it does have a solution, this is
interpreted as evidence in favor of the correctness of the corresponding
equation. The more the number of equations exceeds the number of variables,
the stronger is the evidence that the solution is not just noise but means
something.

The linear systems arising from guessing problems can be solved efficiently
by Hermite-Pad\'e approximation~\cite{beckermann94,giorgi03,BoJeMoSc17,JeaNeiVil2020}.
Modern implementations have no trouble handling examples where $N$ is in
the range of 10000. There are also variants for the multivariate setting~\cite{BeBoFa15,BerthomieuF2018,BerFau20}
as well as experiments with approaches that use machine learning instead of
linear algebra~\cite{gauthier22}. 

Although it rarely happens in practice, a guessed equation may be incorrect.
It is therefore important to be able to assess the quality of a guess. The
amount of overdetermination of the linear system was already mentioned as a
source of trust. Some further tests that may help to distinguish correct
equations from noise have been proposed in~\cite{bostan09}:
a correct equation is likely to contain short integer coefficients
while a wrong guess will typically contain long integers;
a correct recurrence for an integer sequence must produce only integers when it is unrolled
while a wrong guess will typically produce rational numbers;
a correct differential equation for a generating function is likely to have nice singularities
while a wrong guess will typically have awkward singularities;
a correct differential equation for a generating function of an integer sequence with moderate
growth must have nilpotent $p$-curvature~\cite{bostan14,bostan16,pages21}
while a wrong guess will typically not have this property.
All these tests are extremely strong and make guessing a very reliable
tool in practice. Especially for large examples we virtually never encounter
wrong guesses. 

The focus of this paper is on small examples. The assumption is that $N$ is small
and that further terms cannot be obtained at reasonable cost. A prominent example
for such a situation is the number of permutations avoiding the pattern $1324$
for which, despite tremendous efforts~\cite{johansson14,conway18}, only the first
$50$ terms are known. If such a sequence satisfies
an equation of order~$r$ and degree~$d$ with $(r+1)(d+2)>N+2$, we won't be able
to find this equation directly. In this situation, it can be exploited that an equation
of slightly higher order may have substantially lower degree, so that we may find an equation
using $r+1$ and $d/2$, for example. This phenomenon is well understood~\cite{jaroschek13a} and
has been exploited by guessing software since long. Our assumption here is that $N$
is so small that for every choice $(r,d)$ with $(r+1)(d+2)\leq N+2$ the linear system has
no solution, so that trading order against degree does not help. For sequences where
every other term is zero it has been observed~\cite{kauers19b} that removing the zeros from
the data can bring a small advantage. In the present paper, we do not assume that every
other term is zero. 

Our idea is to use the plausibility tests mentioned above in the search for a plausible
candidate equation. For most tests, we do not know how this idea could be reasonably
implemented. For example, restricting the search to recurrences that in addition to
fitting the given data have the property that the next term they generate is an integer
seems to require
solving nonlinear diophantine equations. Similarly, it is not clear how we could
enforce nicely behaved singularities or a nilpotent curvature at a reasonable cost.
The one thing we can do at a reasonable cost is search for equations with short integer
coefficients. The purpose of this paper is to explore this idea. Based on the evidence
reported below, we claim that the search for equations with short integer coefficients
can lead to plausible conjectured equations that are out of reach for classical
guessing algorithms.

\section{Preliminaries}

For a given matrix $M\in\set Z^{n\times m}$, we will need to compute a basis of
$\ker_{\set Z}M$, the $\set Z$-submodule of $\set Z^m$ consisting of all vectors $x\in\set Z^m$
such that $Mx=0$. Recall that this can be done using the Hermite normal form~\cite{cohen93,storjohann96}.
An integer matrix is said to be in Hermite normal form if it has a staircase shape, like
a matrix in row-reduced form, but where the pivot entries may be arbitrary positive integers
(not just~$1$) and the entries above a pivot may be any nonnegative integers smaller than
the pivot. For example,
\[
\begin{pmatrix}
  3 & -5 & 1 & 0 & 2 \\
  0 & 0 & 2 & 0 & -1 \\
  0 & 0 & 0 & 1 & 2
\end{pmatrix}
\]
is a Hermite normal form. For every matrix $M\in\set Z^{n\times m}$, there is a unique
Hermite normal form $H\in\set Z^{n\times m}$, called the Hermite normal form of~$M$,
such that the rows of $H$ generate the same $\set Z$-submodule of $\set Z^n$ as the
rows of~$M$. Moreover, if $H$ is the Hermite normal form of $(M^\top|I_m)$, and if
$r_1,\dots,r_k\in\set Z^m$ are all the nonzero vectors such that $(0,\dots,0,r_i)$
is a row of~$H$, then $\{r_1,\dots,r_k\}$ is a $\set Z$-module basis of $\ker_{\set Z}M$.
For example, for
\[
M = \begin{pmatrix}
  13 & 0 & 18 & 9\\
  1 & 9 & 0 & 0
  \end{pmatrix}
\]
the Hermite normal form of $(M^\top|I_4)$ is
\[
  H=\begin{pmatrix}
  1 & 7 & 7 & 0 & 0 & -10 \\
  0 & 9 & 0 & 1 & 0 & 0 \\
  0 & 0 & 9 & -1 & 0 & -13 \\
  0 & 0\smash{\rlap{\kern.5em\rule[-.5ex]\fboxrule{4.5em}}}
  & 0 & 0 & 1 & -2 
  \end{pmatrix},
\]
so a basis of $\ker_{\set Z}M$ is $\{(9,-1,0,-13),(0,0,1,-2)\}$.

Note that while it makes an essential difference whether we ask for kernel elements
in $\set Q^m$ or kernel elements in $\set Z^m$, it is not an essential difference
whether the entries of $M$ belong to $\set Z$ or to~$\set Q$, because we have $Mx=0$
if and only if $cMx=0$ for every $c\in\set Z\setminus\{0\}$, so we can simply
clear denominators in $M$ if there are any. 

A \emph{reduced basis} of a submodule of $\set Z^m$ is one that consists of relatively
short vectors. The precise definition of ``relatively short'' does not really matter
for our purposes. It suffices to know that there are algorithms, e.g., the LLL algorithm
or the BKZ algorithm,
which turns any given basis into a reduced basis, and that the first vector in a
reduced basis is at most $2^m$ times longer than the shortest nonzero element of the
submodule. See~\cite{cohen93,vzgathen99} for further details and \cite{nguyen09,novocin11,stehle17}
for some recent developments. 

Besides algorithms for finding short vectors in a given module, there are
also general bounds on the lengths of short vectors. A first result in this direction
known as Siegel's lemma~\cite{siegel29} is proved by the pidgeonhole principle. We will
use the following sharper bound, which was achieved by using geometry of numbers~\cite{bombieri83}:

\begin{thm}[Bombieri--Vaaler]\label{thm:BV}
  Let $M\in\set Z^{n\times m}$ with $n<m$, and let $g$ be the $\gcd$ of all
  $n\times n$ minors of~$M$. Then $\ker_{\set Z}M$ contains a nonzero element
  $x\in\set Z^m$ with
  \[
    ||x||_\infty\leq \biggl(\frac1g \sqrt{\det\bigl(MM^\top\bigr)}\,\biggr)^{1/(m-n)}.
  \]
\end{thm}

\section{Algorithms}

For simplicity, we will discuss only the case of guessing linear recurrence equations
with polynomial coefficients. The algorithms are easily adapted to the search for
linear differential equations with polynomial coefficients, or for polynomial equations,
satisfied by the corresponding power series.
Throughout, we consider a $\set Q$-vector space basis $b_0,b_1,\dots$ of the space $\set Q[x]$
of polynomials with the property that for all $j\in\set N$, the polynomials $b_0,\dots,b_j$
generate the subspace of $\set Q[x]$ consisting of all polynomials of degree at most~$j$.

In the classical linear algebra approach, the choice of the basis is irrelevant, so it suffices to
consider the standard basis $b_j=x^j$.
For the variation under consideration here, the choice of the basis may have an effect on the
outcome. We will therefore formulate the algorithms for an arbitrary basis.
In our experiments, we found that the binomial basis $b_j=\binom{x+j}{j}$ worked well, as well
as shifted versions of the standard basis and the binomial basis, e.g., $b_j=(x+\lfloor r/2\rfloor)^j$
or $b_j=\binom{x+\lfloor r/2\rfloor+j}{j}$, where $r$ is the target order of the sought equation.

In order to write the linear system $\sum_{i=0}^r\sum_{j=0}^d c_{i,j} b_j(n) a_{n+i}=0$
for the undetermined coefficients~$c_{i,j}$ in matrix form, we define the following
matrices $A$ and~$B_j$.
For given terms $a_0,\dots,a_N\in\set Q$ and a target order~$r$, the matrix $A$ is defined by
\[
A=\begin{pmatrix}
  a_0 & \cdots & a_r \\ \vdots & \ddots & \vdots \\ a_{N-r} & \cdots & a_N\end{pmatrix}
  \in\set Q^{(N-r+1)\times(r+1)}.
\]
For $j\in\set N$, we further define the matrix $B_j$ by
\[
B_j:=\begin{pmatrix}
  b_j(0) & \cdots & b_j(0) \\ \vdots & & \vdots \\ b_j(N-r) & \cdots & b_j(N-r)
\end{pmatrix}\in\set Q^{(N-r+1)\times(r+1)}.
\]
For two matrices $M_1,M_2$ of the same format, we write $M_1\odot M_2$ for the component-wise product
of $M_1$ and~$M_2$. The space of all linear recurrence equations of order $\leq r$ with polynomial
coefficients of degree $\leq d$ which are valid on the given terms $a_0,\dots,a_N$
is then the kernel of the matrix
\[
  (A\odot B_d|A\odot B_{d-1}|\cdots|A\odot B_0)\in\set Q^{(N-r+1)\times(r+1)(d+1)}.
\]
Our first algorithm computes a short vector with integer components in this space.

\begin{algorithm}\label{alg:1}
  Input: $a_0,\dots,a_N\in\set Q$, $r,d\in\set N$\\
  Output: A linear recurrence of order $r$ and degree $d$ which matches the given data and involves short integers,
  or ``no recurrence found''. 

  \smallskip
  \step 12 Compute a $\set Z$-module basis $v_1,\dots,v_m\in\set Z^{(r+1)(d+1)}$ of
    \[
    \ker_{\set Z} (A\odot B_d|A\odot B_{d-1}|\cdots|A\odot B_0).
    \]
  \step 22 if $m=0$, then return ``no recurrence found''.
  \step 32 Apply LLL to $v_1,\dots,v_m$, call the result $w_1,\dots,w_m$.
  \step 42 Return the recurrence corresponding to the vector~$w_1$.
\end{algorithm}

Note that the output ``no recurrence found'' can only occur if $(r+1)(d+2)\leq N+2$.

\begin{example}
  Consider the sequence $a_n=\sum_{k=0}^nC_k$, where $C_k$ is the $k$th Catalan number.
  The linear algebra approach needs to know $a_0,\dots,a_6$ in order to detect the
  recurrence
  \[
    (6 + 4 n) a_n - (9 + 5 n) a_{n+1} + (3 + n) a_{n+2} = 0.
  \]
  Alg.~\ref{alg:1} can find this equation already from $a_0,\dots,a_5$.
  From these terms and the basis elements $b_0=1$ and $b_1=x$, we construct the matrix
  \[
  M=
  \begin{pmatrix}
    0 & 0 & 0 & 1 & 2 & 4 \\
    2 & 4 & 9 & 2 & 4 & 9 \\
    8 & 18 & 46 & 4 & 9 & 23 \\
    \smash{\rlap{$\underbrace{\kern5.5em}_{=A\odot B_1}$}}27 & 69 & 195 & \smash{\rlap{$\underbrace{\kern4.8em}_{=A\odot B_0}$}} 9 & 23 & 65\\
  \end{pmatrix}.
  \]
  \par\bigskip\bigskip
  \noindent Using the Hermite normal form, it finds that
  \[
    \ker_{\set Z}M=\<\begin{pmatrix}15\\-14\\3\\2\\1\\-1\end{pmatrix},\begin{pmatrix}41\\-37\\8\\0\\12\\-6\\\end{pmatrix}>\subseteq\set Z^6.
  \]
  Applying LLL to this basis gives the reduced basis
  \[
  \{\begin{pmatrix}-4\\5\\-1\\-6\\9\\-3\end{pmatrix},\begin{pmatrix}11\\-9\\2\\-4\\10\\-4\end{pmatrix}\}.
  \]
  The first vector in this basis contains the coefficients of the correct recurrence. 
\end{example}

Alg.~\ref{alg:1} does more than required in that it not only finds one short vector but
a whole basis of short vectors. This can be disadvantageous if the first vector of the
LLL-basis, which is the only one we care about, is much shorter than the other vectors.
In such a situation, it might be better to use the following variant, which is based on
homomorphic images and terminates as soon as the modulus is large enough to recover the
short vector, regardless of how long the remaining vectors are.

\begin{algorithm}\label{alg:2}
  Input/Output: like for Alg.~\ref{alg:1}.

  \smallskip
  \step 12 Let $M=(A\odot B_d|A\odot B_{d-1}|\cdots|A\odot B_0)$.
  \step 22 Let $p$ be a prime and set $q=p$.
  \step 32 Compute a row-reduced basis~$V=\{v_1,\dots,v_m\}$ of $\ker M$ mod~$p$.
    Row-reduced means that the matrix $(v_1|\dots|v_m)^\top$ should be row-reduced.
  \step 42 If $V=\emptyset$, then return ``no recurrence found''.
  \step 52 repeat:
  \step 63 Let $p$ be a new prime and compute a row-reduced basis $W$ of $\ker M$ mod~$p$.
  \step 73 Merge $V$ and $W$ into a basis of $\ker M$ mod $pq$ using Chinese remaindering,
    call this basis $V$ and set $q=pq$.
  \step 83 Apply LLL to $V\cup\{q e_1,\dots, q e_{(r+1)(d+1)}\}$ and let $w$ be the
    first vector in the output, where $e_i$ denotes the $i$th unit vector of
    length $(r+1)(d+1)$.
  \step 92 until the recurrence corresponding to $w$ matches the given data.
  \step {10}2 return this recurrence.  
\end{algorithm}

\begin{thm}
  Alg.~\ref{alg:2} is correct and terminates. 
\end{thm}
\begin{proof}
  Correctness is clear. We show termination. 
  If the algorithm reaches the loop starting in line~5, $\ker M$ is a nontrivial
  subspace of $\set Q^{(r+1)(d+1)}$. Then also the $\set Z$-module $\ker_{\set Z} M\subseteq\set Z^{(r+1)(d+1)}$
  is nontrivial.
  By construction, the lattice generated by $V\cup\{q e_1,\dots, q e_{(r+1)(d+1)}\}$ in line~8
  consists of all vectors $w\in\set Z^{(r+1)(d+1)}$ such that $w$ is in the kernel of $M$ mod~$q$.
  Each such lattice is properly contained in the lattice considered in the previous iteration.
  Therefore, every vector in such a lattice which is not an element of $\ker_{\set Z} M$ will
  disappear in one of the later iterations.
  In particular, for every fixed nonzero element $w$ of $\ker_{\set Z}M$, there will be an
  iteration when all elements of the lattice that do not belong to $\ker_{\set Z}M$ are at least
  $2^{(r+1)(d+1)}$ times longer than~$w$.
  In this situation, LLL is guaranteed to find $w$ or a shorter element of~$\ker_{\set Z}M$. 
\end{proof}

\begin{example}
  Let $a_n$ be defined as in the previous example and suppose again that we want to recover the
  recurrence of order~2 and degree~1 from the known terms $a_0,\dots,a_5$.
  We have
  \[
  \ker_{\set Z_{13}} M = \<\begin{pmatrix}1\\0\\11\\5\\11\\3\end{pmatrix},\begin{pmatrix}0\\1\\6\\8\\8\\7\end{pmatrix}>
    \subseteq\set Z_{13}^6
    \text{ and }
    \ker_{\set Z_{17}} M = \<\begin{pmatrix}1\\0\\9\\14\\6\\2\end{pmatrix},\begin{pmatrix}0\\1\\7\\10\\10\\1\end{pmatrix}>
    \subseteq\set Z_{17}^6,
  \]
  which Chinese remaindering merges to
  \[
  \ker_{\set Z_{221}} M = \<\begin{pmatrix}1\\0\\128\\31\\193\\172\end{pmatrix},\begin{pmatrix}0\\1\\58\\112\\112\\137\end{pmatrix}>
    \subseteq\set Z_{221}^6.
  \]
  Applying LLL to these vectors and $221e_1,\dots,221e_6$ gives the basis
  \[
  \{\begin{pmatrix}-4\\5\\-1\\-6\\9\\-3\end{pmatrix},
    \begin{pmatrix}11\\-9\\2\\-4\\10\\-4\end{pmatrix},
    \begin{pmatrix}0\\8\\22\\12\\12\\-9\end{pmatrix},
    \begin{pmatrix}-34\\-28\\-9\\9\\26\\40\end{pmatrix},
    \begin{pmatrix}15\\12\\-36\\41\\40\\25\end{pmatrix},
    \begin{pmatrix}-37\\-33\\-20\\19\\-8\\-56\end{pmatrix}
  \}.
  \]
  The first vector in this basis corresponds to the correct recurrence.
\end{example}

In most examples we have tried (see Sect.~\ref{sec:nongeneric}),
Alg.~\ref{alg:1} performed better than Alg.~\ref{alg:2}.  This indicates that
in these examples, the generic solutions of the linear systems corresponding
to wrong recurrences are not extremely long compared to the solutions which
correspond to correct recurrences. This proportion however shifts in favor of
Alg.~\ref{alg:2} when $m=\dim\ker M$ is small, i.e., when the provided data is
almost sufficient to succeed with the standard linear algebra guesser.  We
have randomly constructed such examples (see Sect.~\ref{sec:generic}) and
indeed observed a better performance of Alg.~\ref{alg:2} compared to
Alg.~\ref{alg:1}. In addition, the examples we considered had only a
relatively small number~$N$ of known terms, owing to the cost of LLL. It is
typical for sequences arising in combinatorial applications that the size of
the $n$th term grows linearly (or more) in~$n$. Therefore, Alg.~\ref{alg:2}
becomes increasingly interesting for examples in which more (and thus longer)
terms have to be taken into account.

A disadvantage of Alg.~\ref{alg:2} is that LLL has to start from scratch in every iteration.
The reason is that Chinese remaindering must be applied to normalized nullspace bases and
applying LLL destroys the normalization. We do not see how to adapt Alg.~\ref{alg:2} so as
to efficiently recycle the output of one LLL computation in subsequent iterations.

It is possible however to recycle the output of earlier LLL computations if we consider a range of degrees
rather than a single degree. This is detailed in the following algorithm. Note that while in
the linear algebra approach it suffices to consider the largest degree which for a fixed $r$
leads to an overdetermined system, in the LLL-based approach the failure of Alg.~\ref{alg:1}
or Alg.~\ref{alg:2} for a certain degree~$d$ has no immediate implication on whether it will
also fail for degree $d-1$ or degree $d+1$. We therefore need to check a range of degrees. 

\begin{algorithm}\label{alg:3}
  Input: $a_0,\dots,a_N\in\set Z$, $r,d_{\min},d_{\max}\in\set N$, and a function which maps a recurrence candidate
  to True (for ``plausible'') or False (for ``dubious'')\\
  Output: A plausible recurrence of order~$r$ and some degree $d$ in the range $d_{\min},\dots,d_{\max}$ which matches
  the given data and involves short integers, or ``no recurrence found''.

  \smallskip
  \step 12 Set $L=\emptyset$.
  \step 22 for $d=d_{\min},\dots,d_{\max}$, do:
  \step 33 Compute a $\set Z$-module basis $v_1,\dots,v_m\in\set Z^{(r+1)(d+1)}$ of
    \[
    \ker_{\set Z} (A\odot B_d|A\odot B_{d-1}|\cdots|A\odot B_0).
    \]
  \step 43 if $d>d_{\min}$, then:
  \step 54 replace $v_1,\dots,v_m$ by the rows of the Hermite normal form of the
    matrix whose rows are $v_1,\dots,v_m$ and let $\ell$ be such
    that the first $r+1$ components of $v_i$ are zero if and only if $i>\ell$.
  \step 63 attach $r+1$ leading zeros to the vectors in~$L$,
    apply LLL to these vectors and $v_1,\dots,v_\ell$,
    and redefine $L$ to be the output basis.
  \step 73 if the recurrence corresponding to the first vector in $L$ is plausible, then return this recurrence.
  \step 82 return ``no recurrence found''.
\end{algorithm}

Any of the conditions mentioned in the introduction can be used as plausibility
check, for example the condition that the next few terms generated by the candidate
recurrence from the given terms must be integers.

\begin{example}
  Consider now the sequence $a_n=\sum_{k=0}^n C_{3k}$, where $C_k$ is again the $k$th Catalan number.
  This sequence satisfies a recurrence of order~2, degree~3, and integer coefficients
  $|c_{i,j}|\leq3931$, which we want to recover
  from the terms $a_0,\dots,a_7$ using Alg.~\ref{alg:3}.
  For $d=0$ and $d=1$, the linear systems are overdetermined and their solution spaces are~$\{0\}$.
  For $d=2$, we find that $\ker_{\set Z}M$ is generated by three vectors, and applying LLL to them
  gives 
  \[
  \begin{pmatrix} -10770777\\ 12213849\\ -188031\\ 34593817\\ -38801808\\ 316322\\ 49496244\\ -55183749\\ 2040625 \end{pmatrix},\quad
  \begin{pmatrix} -5821934\\ 49822837\\ -776612\\ -57412313\\ 38007895\\ -1750361\\ 9044412\\ -14700087\\ 573595 \end{pmatrix},\quad
  \begin{pmatrix} 89543113\\ -3923255\\ 39847\\ -24321368\\ -8364412\\ 157455\\ 5994312\\ -6286292\\ 229880 \end{pmatrix}.
  \]
  For $d=3$, we find that $\ker_{\set Z}M$ is generated by six vectors. Using the Hermite normal form,
  we find the following basis:
  \[
  \begin{pmatrix} 1\\ 0\\ 5\\ 0\\ 2\\ \rule{.5ex}{6ex} \end{pmatrix}, \
  \begin{pmatrix} 0\\ 1\\ 2\\ 0\\ 1\\ \rule{.5ex}{6ex} \end{pmatrix}, \
  \begin{pmatrix} 0\\ 0\\ 9\\ 0\\ 1\\ \rule{.5ex}{6ex} \end{pmatrix}, \ 
  \begin{pmatrix} 0\\ 0\\ 0\\ 1\\ 1\\ \rule{.5ex}{6ex} \end{pmatrix}, \
  \begin{pmatrix} 0\\ 0\\ 0\\ 0\\ 3\\ \rule{.5ex}{6ex} \end{pmatrix}, \
  \begin{pmatrix} 0\\ 0\\ 0\\ 0\\ 0\\ \rule{.5ex}{6ex} \end{pmatrix},
  \]
  where each black rectangle hides seven integers with 19 or more decimal digits. 
  We replace the last three of these vectors by the LLL-basis of the previous operation,
  prolonged by three leading zeros, and call LLL.
  The first vector in the resulting basis contains the coefficients of the correct recurrence.
\end{example}

As the specification of Alg.~\ref{alg:3} is intentionally vague with regard to the shortness of the
integers in the output, we refrain from making a formal correctness statement. Instead, we only show
that the recycling of reduced bases from earlier iterations is done correctly.

\begin{thm}
  In the $d$th iteration of Alg.~\ref{alg:3}, the vectors to which LLL is applied in step~6
  form a $\set Z$-module basis of $\ker_{\set Z}(A\odot B_d|\cdots|A\odot B_0)$.
\end{thm}
\begin{proof}
  By induction on~$d$. For $d=d_{\min}$, we have $L=\emptyset$, and none of the $v_1,\dots,v_m$ are
  discarded due to the if-condition in line~4, therefore in line~6 LLL is applied to the vectors
  $v_1,\dots,v_m$ which were chosen as a basis of $\ker_{\set Z}(A\odot B_d|\cdots|A\odot B_0)$ in line~3.
  
  If for some $d>d_{\min}$ the claim is true in iteration $d-1$, then at the beginning of the $d$th iteration, $L$~contains
  a basis of $\ker_{\set Z}(A\odot B_{d-1}|\cdots|A\odot B_0)$. Denote the elements of this basis, padded
  with $r+1$ zeros, by $w_1,\dots,w_k$. Let $v_1,\dots,v_m$ form a basis of $\ker_{\set Z}(A\odot B_{d}|\cdots|A\odot B_0)$
  and at the same time the rows of a Hermite normal form, as in step~5. We then have to show that
  the $\set Z$-modules $\<w_1,\dots,w_k>$ and $\<v_{\ell+1},\dots,v_m>$ are equal.
  
  ``$\supseteq$'': As each $v_i$ ($i\in\{\ell+1,\dots,m\}$) is an element of
  $\ker_{\set Z}(A\odot B_d|\cdots|A\odot B_0)$ with $r+1$ leading zeros, chopping the $r+1$ leading zeros
  turns it into an element of $\ker_{\set Z}(A\odot B_{d-1}|\cdots|A\odot B_0)$.
  By assumption on $w_1,\dots,w_k$, it follows that $v_i\in\<w_1,\dots,w_k>$.   
  
  ``$\subseteq$'': As each $w_i$ ($i=1,\dots,k$) is an element of
  \[
  \ker_{\set Z}(A\odot B_{d-1}|\cdots|A\odot B_0)
  \]
  padded with $r+1$ leading zeros, it is an element of $\ker_{\set Z}(A\odot B_d|\cdots|A\odot B_0)$ and
  hence a $\set Z$-linear combination of $v_1,\dots,v_m$, say $w_i=\alpha_1v_1+\cdots+\alpha_mv_m$. It remains to
  show that $\alpha_1=\cdots=\alpha_i=0$. To see this, observe that $0=\pi(w_i)=\alpha_1\pi(v_1)+\cdots+\alpha_i\pi(v_i)+0$,
  where $\pi\colon\set Z^{(r+1)(d+1)}\to\set Z^{r+1}$ is the projection to the first $r+1$ coordinates. As
  $v_1,\dots,v_m$ form the rows of a Hermite normal form, so do $\pi(v_1),\dots,\pi(v_i)$, and as these vectors
  are nonzero and the nonzero rows of a Hermite normal form are linearly independent, we have $\alpha_1=\cdots=\alpha_i=0$,
  as claimed. 
\end{proof}

The LLL-algorithm has an incremental nature. In order to compute a reduced basis from an input basis $v_1,\dots,v_m$,
it first computes a reduced basis for the input basis $v_1,\dots,v_{m-1}$ recursively and then adjusts this reduced
basis to a reduced basis for $v_1,\dots,v_m$. In particular, if we know that the first few vectors of $v_1,\dots,v_m$
are already reduced, as in Alg.~\ref{alg:3}, the LLL algorithm can take this information into account. It is therefore
possible to implement Alg.~\ref{alg:3} in such a way that the total cost of all the LLL computations combined is no
more than the cost of a single LLL computation applied to a basis of $\ker_{\set Z}(A\odot B_{d_{\max}}|\cdots|A\odot B_0)$
whose elements form the rows of a Hermite normal form. 
This idea is also used in van Hoeij's LLL-based factorization algorithm~\cite{hoeij13}. 

An implementation of Alg.~\ref{alg:3} has been added to the Guess.m package for Mathematica~\cite{kauers09a}.
It is available as function \verb|GuessZUnivRE|. We will also include the algorithm in
the ore\_algebra package for Sage~\cite{kauers14b}.

\section{The Generic Case}\label{sec:generic}

In the linear algebra approach, the condition $(r+1)(d+2)\leq N+2$ imposes a restriction
on the choices of $r$ and $d$ in dependence on~$N$. By allowing underdetermined linear systems
in the new approach, we can also explore larger choices of $r$ and~$d$.
The underdetermined linear systems in the new approach will in general have solutions corresponding
to correct equations as well as solutions corresponding to incorrect equations. 
The limiting factor for the choice of $r$ and $d$ is now the requirement
that the correct equations must
have shorter coefficients than the incorrect equations, so that LLL has a chance to tell them
apart. This limiting factor is harder to quantify. In this section, we offer a somewhat
heuristic analysis of when this can be expected in the ``generic'' situation.

Assume that $\sum_{i=0}^r\sum_{j=0}^d c_{i,j} n^j a_{n+i}=0$ is a random recurrence
whose coefficients $c_{i,j}$ and initial values $a_0,\dots,a_{r-1}$ have bitsize
$\ell+1$, which means that they are randomly (uniformly) chosen from the set
$\{-2^\ell+1,\dots,2^\ell\}\subset\set{Z}$. By rewriting the recurrence as
\[
  a_{n+r} = -\frac{1}{p_r(n)}\cdot\bigl(p_{r-1}(n)a_{n+r-1} + \dots + p_0(n)a_n\bigr),
\]
where $p_i(n)=\sum_{j=0}^d c_{i,j}n^j$, we are able to unroll it. Clearly, in
such a generic situation, the numerators and denominators of $|a_n|$ will
grow, because there is no reason to expect considerable cancellations other
than accidentally occurring small common factors.
By defining
\[
  u_{n} = a_{r-1}\prod_{i=0}^{n-r} p_{r-1}(i),\quad
  v_{n} = \prod_{i=0}^{n-r} \bigl(-p_{r}(i)\bigr),
\]
we can approximate the $n$-th sequence term~$a_n$ by the quotient $u_n/v_n$,
and we may assume that $\gcd(u_n,v_n)$ is negligibly small compared to the
absolute values of $u_n$ and~$v_n$. Hence, this approximation omits the lower
terms $a_n,\dots,a_{n+r-2}$ in the recurrence equation, which corresponds to
approximating it by its two leading terms $p_r(n)a_{n+r}+p_{r-1}(n)a_{n+r-1}=0$.
Nevertheless, $v_n$ is a common denominator of all the sequence terms $a_0,\dots,a_n$.

Assume that the matrix $(A\odot B_d|A\odot B_{d-1}|\cdots|A\odot B_0)$ from
the guessing problem has $k:=N-r+1$ rows and $m:=(r+1)(d+1)$ columns, and that
the latter are indexed by pairs $(s,t)$ with $0\leq s\leq r$ and $0\leq t\leq
d$. No particular order on the pairs $(s,t)$ is specified here, since it will
be irrelevant for the subsequent analysis.

Since the sequence $(a_n)$ contains rational numbers, we must ensure that the
matrix has integer entries. This is achieved by noting that $a_{i+r}$
is the sequence element with the highest index that appears in row~$i$,
and hence multiplying this row by $v_{i+r}$ will do the job:
\[
  M = \diag(v_r,\dots,v_{r+k-1})\cdot(A\odot B_d|\cdots|A\odot B_0)\in\set{Z}^{k\times m}.
\]
The entries of~$M$ are then $m_{i,(s,t)}=i^t a_{i+s} v_{i+r}$, and therefore
the entries of the matrix $MM^\top$ are given by
\[
  \bigl(MM^{\!\top}\bigr)_{\!i,j} = \sum_{(s,t)} m_{i,(s,t)} m_{j,(s,t)} =
  \sum_{(s,t)} (ij)^t a_{i+s} a_{j+s} v_{i+r} v_{j+r}.
\]
The summand of this sum is expected to take its largest (absolute) value for
$s=r$ and $t=d$ (or $t=0$ when $ij=0$). Hence we approximate the sum by its
largest, dominating part (assuming $ij\neq0$)
\[
  (ij)^d a_{i+r} a_{j+r} v_{i+r} v_{j+r} = (ij)^d u_{i+r} u_{j+r}.
\]
Because the latter expression has the form $g(i)g(j)$ with $g(i)=i^du_{i+r}$,
it follows that all products $\prod_{i=0}^{k-1}
\bigl(MM^\top\bigr)_{i,\sigma(i)}$ are approximately of the same size, where
$\sigma$ runs through all permutations of $\{0,\dots,k-1\}$. This shows that
$\det\bigl(MM^\top\bigr)$ can hardly exceed $k!\prod_{i=0}^{k-1}g(i)^2$, while
due to cancellations, it may get arbitrarily close to~$0$. A reasonable
balance between the two extreme cases is to estimate the determinant as
follows:
\[
  \det\bigl(MM^\top\bigr)
  \approx \prod_{i=0}^{k-1} \bigl(MM^\top\bigr)_{i,i}
  \approx ((k-1)!)^{2d} \prod_{i=0}^{k-1} u_{i+r}^2.
\]
Finally, we need to estimate the size of the quantities~$u_n$. This is
done by replacing the polynomial~$p_{r-1}(i)$ in the definition 
\[
  u_n = a_{r-1}\prod_{i=0}^{n-r} p_{r-1}(i)
\]
by its leading term (except when $i=0$):
\[
  u_n \approx a_{r-1} \, c_{r-1,0} \prod_{i=1}^{n-r} c_{r-1,d} \, i^d =
  a_{r-1} \, c_{r-1,0} \, c_{r-1,d}^{n-r} \, ((n-r)!)^d.
\]
This allows us to estimate and simplify as follows
\begin{align*}
  \prod_{i=0}^{k-1} u_{i+r} &\approx
  \prod_{i=0}^{k-1} a_{r-1} \, c_{r-1,0} \, c_{r-1,d}^i \, (i!)^d \\
  &\approx (a_{r-1} \, c_{r-1,0})^k c_{r-1,d}^{k(k-1)/2} \cdot
  \left(\prod_{i=1}^{k-1} i!\right)^{\!d}.
\end{align*}
Combining everything and ignoring the gcd~$g$ from Theorem~\ref{thm:BV}
(which is hard to predict and in many instances just equal to~$1$),
our approximation for the Bombieri-Vaaler bound is
\begin{multline*}
  \sqrt{\det\bigl(MM^\top\bigr)}^{1/(m-k)} \approx \\
  \left((a_{r-1} \, c_{r-1,0})^k c_{r-1,d}^{k(k+1)/2}
  \left(\frac1k\prod_{i=1}^k i!\right)^{\!\!d} \,
  \right)^{\!\!1/((r+1)(d+1)-k)}.
\end{multline*}
We test experimentally how good this approximation is. For about $1200$
randomly generated recurrences, where $1\leq r\leq6$, $0\leq d\leq 6$, and
$6\leq\ell\leq100$, we plot the ratio between the bitsize of our
approximation and the bitsize of the true value of $\sqrt{\det(MM^\top)}^{1/(m-k)}$;
the ratios are sorted in increasing order.
In the vast majority of the cases, we are off by less than~$5\%$.

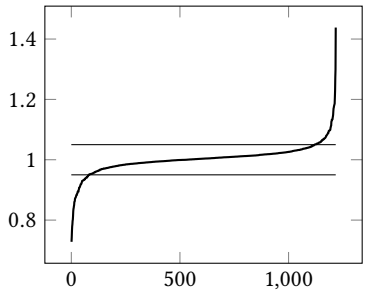
\begin{center}
  \begin{tikzpicture}
    \begin{axis}[width=5.8cm]
      \addplot[thin,black] coordinates {(0,1.05) (1217,1.05)};
      \addplot[thin,black] coordinates {(0,.95) (1217,.95)};
      \addplot[thick,black] table[col sep=comma]{BV-estimate.csv};
    \end{axis}
  \end{tikzpicture}
\end{center}

The question now is, for which~$k$, depending on $r,d,\ell$, the above
expression becomes smaller than~$2^\ell$. In that case we should not be able
to identify the true solution, because the Bombieri-Vaaler bound predicts a
solution with coefficients smaller than those of the sought recurrence.
Taking into account that $a_{r-1},c_{r-1,0},c_{r-1,d}$ are all bounded
by~$2^\ell$, and using the asymptotic expansion of the hyperfactorial, our
approximation leads to the following inequality
\[
  2^{\ell k(k+7)/2} \cdot
  \left(\frac{k^{k^2/2+k-7/12}}{\e^{3k^2/4+k-1/12}} \,
  \bigl(\sqrt{2\pi}\bigr)^{k+1}\right)^{\!d} \lesssim 2^{\ell(r+1)(d+1)}.
\]
For $\ell\to\infty$ the term in parentheses becomes insignificant, so that
after omitting it we can solve the inequality $k(k+7) \lesssim 2(r+1)(d+1)$
explicitly and obtain the final ``soft'' bound, in terms of $N=k+r-1$:
\[
  N \lesssim \frac12\Bigl(\sqrt{8(r+1)(d+1)+49}-7\Bigr)+r-1.
\]
It signifies that it is unlikely to identify the correct recurrence when
$N,r,d$ satisfy the inequality. However, since it was obtained by rough
estimates, it does not allow us to draw any definite statements in form of a
theorem.

The following experiments show that nevertheless our bound yields quite accurate
predictions: we randomly generated recurrences for $r=4,8$ and $d=0,\dots,6$
with integer coefficients $|c_{i,j}|<2^{16}$. The plots show the smallest~$N$
for which the recurrence could be guessed from $a_0,\dots,a_N$ (depicted as
dots), and the graph of the bound (depicted as a line):

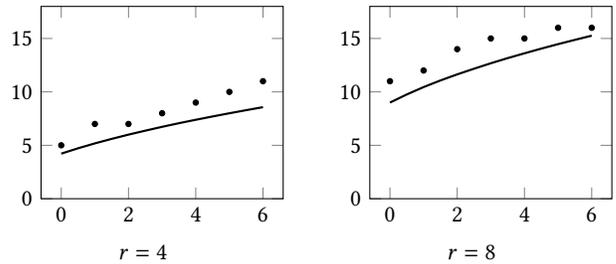
\begin{center}
  \begin{tabular}{@{}c@{\qquad}c@{}}
  \begin{tikzpicture}
    \begin{axis}[width=4.8cm,ymin=0,ymax=18]
      \addplot[thick,black] coordinates {%
        (0., 4.22) (0.5, 4.72) (1., 5.18) (1.5, 5.6) (2., 6.) (2.5, 6.37) (3., 6.73)
        (3.5, 7.07) (4., 7.39) (4.5, 7.7) (5., 8.) (5.5, 8.29) (6., 8.57)};
      \addplot[only marks, mark size=1pt] coordinates {%
        (0,5) (1,7) (2,7) (3,8) (4,9) (5,10) (6,11)};
    \end{axis}
  \end{tikzpicture}
  &
  \begin{tikzpicture}
    \begin{axis}[width=4.8cm,ymin=0,ymax=18]
      \addplot[thick,black] coordinates {%
        (0., 9.) (0.5, 9.76) (1., 10.45) (1.5, 11.07) (2., 11.64) (2.5, 12.17) (3., 12.68)
        (3.5, 13.16) (4., 13.61) (4.5, 14.05) (5., 14.47) (5.5, 14.87) (6., 15.26)};
      \addplot[only marks, mark size=1pt] coordinates {%
        (0,11) (1,12) (2,14) (3,15) (4,15) (5,16) (6,16)};
    \end{axis}
  \end{tikzpicture}
  \\
  $r=4$ & $r=8$
  \end{tabular}
\end{center}

\section{The Non-Generic Case}\label{sec:nongeneric}

Sequences arising in applications cannot simply be assumed to behave like generic
sequences. In order to get some idea how our method performs in practice, we have
evaluated it experimentally. The Online Encyclopedia of Integer Sequences (OEIS,~\cite{sloane})
meanwhile contains more than 350000 sequences arising from all kinds of different contexts.
Salvy~\cite{salvy05} estimated in 2005 that up to 25\% of the entries in the OEIS are D-finite.
We have gone through the entire database and determined all the sequences for which
at least 50 terms were given and for which linear algebra can find a recurrence using
no more than 250 terms. Cases where the linear algebra guesser already succeeded with 10
terms or less were discarded from consideration (too simple). The search resulted in
some 6700 hits. For each of these sequences, we determined the minimal number of terms
needed by a guesser based on linear algebra in order to detect a recurrence as well as
the minimal number of terms needed by an LLL-based guesser in order to detect a recurrence
that is consistent with the one found by the linear-algebra based guesser. 

The LLL-based guesser was called with the four bases $x^n$, $(x+\lfloor r/2\rfloor)^n$,
$\binom{x+n}{n}$, and $\binom{x+\lfloor r/2\rfloor+n}{n}$. Each of these bases beats
all the others on at least some examples, although there is a clearly visible trend
that $(x+\lfloor r/2\rfloor)^n$ wins most of the time. For our evaluation, we determined
the minimal number of terms of a sequence for which the LLL-based guesser recognizes
the recurrence for at least one of the bases. If $N_{\text{linalg}}$ is the minimal
number of terms needed by the linear algebra guesser and $N_{\text{LLL}}$ the minimal
number of terms needed by the LLL-based guesser with at least one of the four bases,
the following picture shows the quotients $N_{\text{LLL}}/N_{\text{linalg}}$ for all
the 6700 sequences taken from the OEIS, in increasing order.

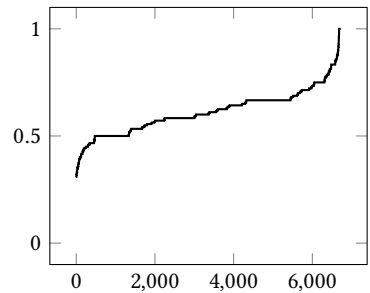
\begin{center}
  \begin{tikzpicture}
    \begin{axis}[width=5.8cm,ymin=-.1,ymax=1.1]
      \addplot[thick,black] table[col sep=comma]{saving.csv};
    \end{axis}
  \end{tikzpicture}
\end{center}

On the average, the LLL-guesser needs 61.2\% of the terms needed by the linear algebra
guesser. For more than 99\% of the sequences, we were able to save at least one term.
The whole computation using our Mathematica implementation of Alg.~\ref{alg:3} took roughly
one CPU year.

Equations in the test set of 6700 sequences are not equally large, and different sizes
are far from equally distributed. The following picture shows the numbers $N_{\text{linalg}}$
in increasing order, indicating that almost all the sequences have rather short equations.

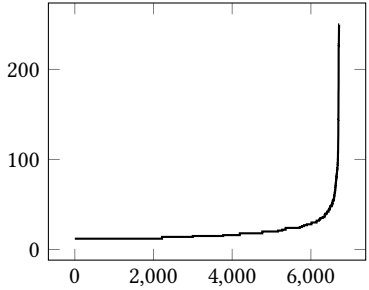
\begin{center}
  \begin{tikzpicture}
    \begin{axis}[width=5.8cm]
      \addplot[thick,black] table[col sep=comma]{sizes.csv};
    \end{axis}
  \end{tikzpicture}
\end{center}

The strong dominance of short equations induces a bias into the saving statistics shown
before. If we restrict the attention to sequences with $N_{\text{linalg}}>50$, we see
a significant saving is also possible for these cases. The mean for this restricted
sample is 64.6\%, and for 12.4\% of these cases no saving was possible at all. In view
of the analysis of Sect.~\ref{sec:generic}, we expected that more saving is possible
for larger equations than on small equations. 

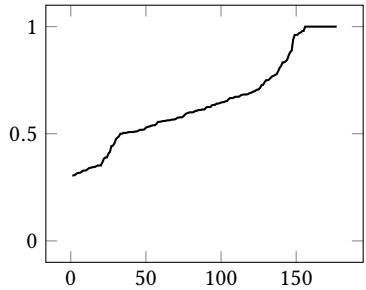
\begin{center}
  \begin{tikzpicture}
    \begin{axis}[width=5.8cm,ymin=-.1,ymax=1.1]
      \addplot[thick,black] table[col sep=comma]{saving-large.csv};
    \end{axis}
  \end{tikzpicture}
\end{center}

The next picture shows the bias towards small equations in the test set from a different
perspective. We show here all the
points $(N_{\text{linalg}},N_{\text{LLL}})$. The number of points in the picture is far
less than 6700 because many sequences share the same point. Taking multiplicities into
account, the cloud's center of gravity is $(19.4, 12.0)$. 

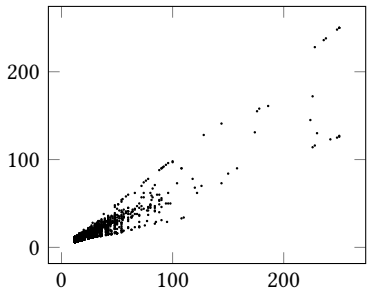
\begin{center}
  \begin{tikzpicture}
    \begin{axis}[width=5.8cm]
      \addplot[only marks, mark size=.25pt] table[col sep=comma]{scatter.csv};
    \end{axis}
  \end{tikzpicture}
\end{center}

\section{Applications}

We have also tried to find something new with LLL-based guessing. 
About 7500 of the sequences in the OEIS are labeled by the keyword `hard', which is supposed
to indicate that a substantial amount of work would be needed to compute further
terms. Usually this simply means that the best known algorithm for computing
the $n$th term of the sequence needs time exponential in~$n$. Many of these hard
sequences are defined through a number-theoretic property and are unlikely to be
D-finite. If we discard from the 7500 hard sequences those in which the word `prime'
appears in the description, we are left with 3764 sequences. 545 of them
have at least 30 and at most 150 known terms.

Classical guessing does not find a recurrence or differential equation for any of
these sequences.
We applied our LLL-based approach for all $r,d$ with $(r+1)(d+2)\leq 3N$, where
$N$ is the number of known terms.
A recurrence was automatically rejected if the next ten terms produced by the
recurrence from the known terms were not all integers.
For 37 sequences, a recurrence was found that passed this test.
These recurrences were inspected by hand, and in most cases, we remain skeptic
about their general validity. For example, if the recurrence we found for A054500
were correct, all sequence terms for $n\geq30$ would be equal to~29, which does
not seem right. 

At least in the following two cases, we believe that the guessed recurrence is correct.

The $n$th term of the sequence A307717 is defined as the number of palindromic integers
(base~10) whose square is palindromic as well and has $n$ decimal digits. The sequence
begins with $4$,~$0$, $2$, $0$, $5$, $0$,~$3$, and the OEIS contains 70 terms. With these terms, we
were able to recognize a recurrence of order~6 and degree~9 and small coefficients. The
guessed recurrence admits a quasi-polynomial solution, suggesting the closed form
(valid for $n>1$)
\[
 a_n=\left\{\begin{array}{ll}
 0&\text{ if $n=0\bmod2$}\\[3pt]
 \tfrac1{192}(195+203n-15n^2+n^3)&\text{ if $n=1\bmod4$}\\[3pt]
 \tfrac1{384}(501+107n-9n^2+n^3)&\text{ if $n=3\bmod4$}
 \end{array}\right.
\]
The OEIS entry for the related sequence A218035 contains a proof of this formula,
which indicates that A307717 is misclassified as `hard' and confirms that our guessed
recurrence was indeed correct.

A more interesting case is the sequence A189281, whose $n$th term is defined as the
number of permutations $\pi\in S_n$ such that $\pi(i+2)-\pi(i)\neq2$ for $i=1,\dots,n-2$.
The sequence appears in~\cite{kotesovec13} and begins with $1$,~$1$, $2$, $5$,~$18$, and the OEIS
contains 36 terms. With these terms, we were able to identify a convincing recurrence
of order~10 and degree~6 and small coefficients. We find this guess convincing for
several reasons:
\begin{itemize}
\item The recurrence can also be detected if we only supply 29 of the 36 known
  terms to the guesser. 
\item If $(c_n)_{n\geq0}$ denotes the sequence defined by the initial
  values of A189281 and the guessed recurrence, then $c_{36},\dots,c_{100}$ turn out
  to be integers.
\item If we put the terms $c_0,\dots,c_{100}$ into a classical guesser, it
  finds a recurrence of order~8 with polynomial coefficients of degree~11.
  This means that the guessed order-10 operator has a right factor of order~8, a
  feature that is not to be expected for a random operator.
\item The exponential generating function of the sequence $(c_n)$ satisfies a
  differential equation that only has regular singularities. 
\item With the Mathematica code provided in the OEIS, we have made an effort
  to compute four more terms of the sequence A189281, and we found them
  to match $c_{36},c_{37},c_{38},c_{39}$, meaning that four terms not used
  for guessing were correctly predicted by the guessed equation. 
\end{itemize}
\begin{conjecture}
  The sequence A189281 satisfies a recurrence of order~8 and degree~11.
\end{conjecture}
We have not tried to prove this conjecture.

\section{Conclusion}

We have demonstrated that a linear recurrence equation can be guessed by using
significantly fewer data compared to classical approaches, by imposing certain
restrictions on its integer coefficients. Driving this idea to its limit, we can
ask how many terms are really necessary for finding a recurrence. There is evidence
that just a single term can be enough. We illustrate this with an example.

\begin{example}
  The central Delannoy numbers $D_n=\sum_{k=0}^n \binom{n}{k} \binom{n+k}{k}$
  satisfy a recurrence of order~$2$ and degree~$1$:
  \[
    (n+2) D_{n+2} - (6n+9) D_{n+1} + (n+1) D_n = 0.
  \]
  We found that this recurrence can be recovered from the single sequence term
  $D_8=265729$. More precisely, we exhaustively enumerated all recurrences of
  order~$2$ and degree~$1$, that have integer coefficients $\leq9$ in absolute
  value, and combined them with all possible initial values $0\leq a_0\leq9$
  and $0\leq a_1\leq9$, yielding $19^6\cdot 10^2=4{,}704{,}588{,}100$ sequences
  in total. Of course, some recurrences can be immediately discarded, because
  \begin{itemize}
  \item e.g., the integer coefficients have a nontrivial gcd,
  \item or the leading coefficient vanishes for some $n\in\set{N}_0$,
  \item or the leading / trailing coefficient is zero (order drop).
  \end{itemize}
  From the remaining cases, we selected those whose sequence terms
  $a_2,\dots,a_{20}$ were all integral. Then, for $n=2,3,\dots,20$, we
  recorded the numbers of sequences that agreed with the sequence~$D_n$ at
  positions~$n$, i.e., $a_n=D_n$, yielding the following statistics:
  \[
  \begin{array}{r|cccccccccc}
    n & 2 & 3 & 4 & 5 & 6 & 7 & 8 & 9 & \cdots & 20 \\ \hline
    \rule{0pt}{12pt}%
    \# & 85994 & 48240 & 3056 & 853 & 25 & 8 & 1 & 1 & \cdots & 1
  \end{array}
  \]
  For example, there were $3056$ sequences $(a_n)$ which had $a_4=D_4=321$ (but
  not necessarily $a_0=D_0$ etc.), and only a single sequence that agreed at
  $n=8$. Hence, the above recurrence is uniquely determined by the single
  sequence term~$D_8$, under the given restrictions on the size of the
  coefficients and initial values. This toy example was solved by brute force;
  we are not aware of any algorithm that could solve such problems
  efficiently in general.

  While at first glance, the index $n=8$ in our example appears to be extremely
  low, this phenomenon can be understood in general by relating the number of bits necessary
  to encode the recurrence with the bitsize of the $n$th sequence term. But even
  if the latter exceeds the former, there is no guarantee that
  a unique solution exists --- this would require to decide some variant of the
  Skolem problem, which already for C-finite recurrences is still unsolved.
\end{example}

\noindent\textbf{Acknowledgement.}
We thank the referees for their careful reading and their valuable suggestions.

\bibliographystyle{plain}
\bibliography{bib}

\end{document}